\documentclass[sigconf]{acmart}

\AtBeginDocument{%
  }

\copyrightyear{2023} 
\acmYear{2023} 
\setcopyright{acmlicensed}\acmConference[WWW '23]{Proceedings of the ACM Web Conference 2023}{April 30-May 4, 2023}{Austin, TX, USA}
\acmBooktitle{Proceedings of the ACM Web Conference 2023 (WWW '23), April 30-May 4, 2023, Austin, TX, USA}
\acmPrice{15.00}
\acmDOI{10.1145/3543507.3583313}
\acmISBN{978-1-4503-9416-1/23/04}

\acmSubmissionID{3008}

\graphicspath{{./figures/}}
\usepackage{subfigure}

\usepackage{algorithm}
\usepackage{algorithmic}

\newtheorem{assumption}{Assumption}[section]

\begin{document}

\title[RL-MPCA]{RL-MPCA: A Reinforcement Learning Based Multi-Phase Computation Allocation Approach for Recommender Systems}



\author{Jiahong Zhou}
\authornote{Corresponding author}
\affiliation{%
 \institution{Meituan, Beijing, China}
 \city{}
 \country{}
}
\email{zhoujiahong02@meituan.com}

\author{Shunhui Mao}
\affiliation{%
 \institution{Meituan, Beijing, China}
 \city{}
 \country{}
}
\email{maoshunhui@meituan.com}

\author{Guoliang Yang}
\affiliation{%
 \institution{Meituan, Beijing, China}
 \city{}
 \country{}
}
\email{yangguoliang@meituan.com}

\author{Bo Tang}
\affiliation{%
 \institution{Meituan, Beijing, China}
 \city{}
 \country{}
}
\email{tangbo17@meituan.com}

\author{Qianlong Xie}
\affiliation{%
 \institution{Meituan, Beijing, China}
 \city{}
 \country{}
}
\email{xieqianlong@meituan.com}

\author{Lebin Lin}
\affiliation{%
 \institution{Meituan, Beijing, China}
 \city{}
 \country{}
}
\email{linlebin@meituan.com}

\author{Xingxing Wang}
\affiliation{%
 \institution{Meituan, Beijing, China}
 \city{}
 \country{}
}
\email{wangxingxing04@meituan.com}

\author{Dong Wang}
\affiliation{%
 \institution{Meituan, Beijing, China}
 \city{}
 \country{}
}
\email{wangdong07@meituan.com}

\renewcommand{\shortauthors}{Jiahong Zhou et al.}

\begin{abstract}
  
  Recommender systems aim to recommend the most suitable items to users from a large number of candidates. Their computation cost grows as the number of user requests and the complexity of services (or models) increases.
  Under the limitation of computation resources (CRs), how to make a trade-off between computation cost and business revenue becomes an essential question. 
  The existing studies focus on dynamically allocating CRs in queue truncation scenarios (i.e., allocating the size of candidates), and formulate the CR allocation problem as an optimization problem with constraints. Some of them focus on single-phase CR allocation, and others focus on multi-phase CR allocation but introduce some assumptions about queue truncation scenarios. However, these assumptions do not hold in other scenarios, such as retrieval channel selection and prediction model selection. Moreover, existing studies ignore the state transition process of requests between different phases, limiting the effectiveness of their approaches. 

  This paper proposes a Reinforcement Learning (RL) based Multi-Phase Computation Allocation approach (RL-MPCA), which aims to maximize the total business revenue under the limitation of CRs. RL-MPCA formulates the CR allocation problem as a Weakly Coupled MDP problem and solves it with an RL-based approach. Specifically, RL-MPCA designs a novel deep Q-network to adapt to various CR allocation scenarios, and calibrates the Q-value by introducing multiple adaptive Lagrange multipliers (adaptive-$\lambda$) to avoid violating the global CR constraints.
  Finally, experiments on the offline simulation environment and online real-world recommender system validate the effectiveness of our approach.
\end{abstract}




\maketitle

\section{Introduction \label{sec:introduction}}
Recommender systems aim to recommend the most suitable items 
to users from a large number of candidates and expect to gain revenue from users' views, clicks, and purchases. They are playing an increasingly important role in e-commerce platforms \cite{hussien2021recommendation}.

\begin{figure}[htbp]
  \centering
  \includegraphics[width=0.9\linewidth]{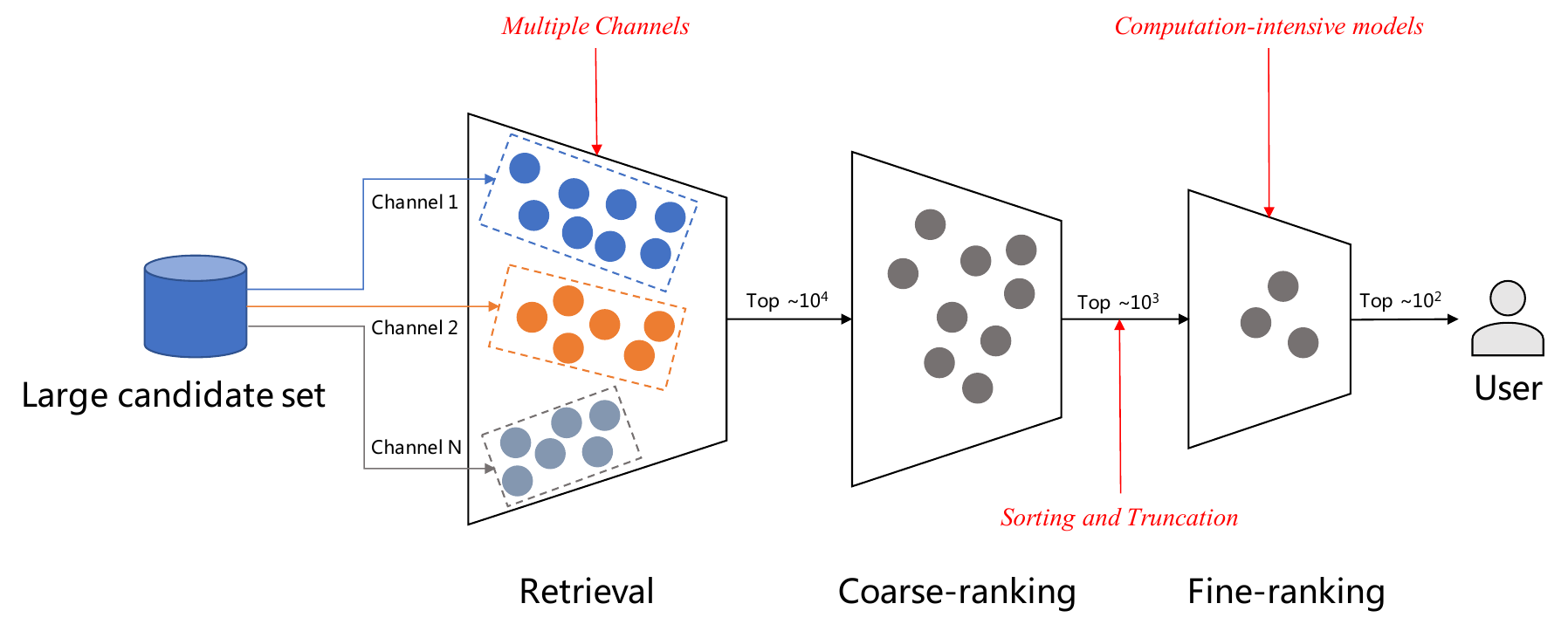}
  \caption{The typical structure of recommender systems.}
  \Description{}
  \label{fig:typical-recommender-system}
\end{figure}

Industrial recommender systems are often designed as cascading architectures \cite{covington2016deep, liu2017cascade}. 
As shown in Figure \ref{fig:typical-recommender-system}, a typical recommender system consists of several stages, including retrieval, coarse-ranking, fine-ranking, etc. 
In these stages, online advertising systems (a kind of recommender system applied to online advertising) generally contain several computation-intensive services or models, including bid models \cite{yang2020motiac, he2021unified}, prediction services \cite{zhou2018deep, feng2019deep}, etc. These services require a lot of computation resources\footnote[1]{In general, computation resources include CPU/GPU computing capacity, memory capacity and response time, etc.} (CRs). 
Take the display advertising system of Meituan Waimai platform\footnote[2]{\url{https://waimai.meituan.com/}, one of the largest e-commerce platforms in China.} (hereinafter referred to as Meituan advertising system), for example. It consumes a lot of CRs in both the retrieval stage and the fine-ranking stage.
As the number of user requests increases dramatically, the system's CR consumption rises accordingly. 
Due to the limitation of CRs, recommender systems need to make a trade-off between CR cost and business revenue when the traffic exceeds the system load. 
From the perspective of CR utilization efficiency, the goal of recommender systems is to maximize the total business revenue under the CR constraint.

To address the challenges of huge traffic and a large number of candidate items, the real-world recommender systems usually use two types of strategies: static strategies and dynamic strategies \cite{jiang2020dcaf, yang2021computation}. 
Static strategies select suitable fixed rules through stress testing and practical experience to allocate CRs. 
They also provide fixed downgrades to cope with unexpected traffic. 
Static strategies require constant manual intervention to adapt to quick changes in traffic, and fixed downgrades provided by static strategies are generally detrimental to business revenue and user experience. 
Dynamic strategies \cite{jiang2020dcaf, yang2021computation} dynamically allocate CRs for requests based on the value of requests.
They prioritize allocating CRs to more valuable requests to achieve better revenue. Compared to static strategies, dynamic strategies are more efficient in utilizing CRs and require fewer manual intervention. 

Recommender systems with multiple stages have various CR allocation scenarios. Based on the application scenario, we summarize the dynamic CR allocation methods into three types: \textbf{Elastic Channel}, \textbf{Elastic Queue}, and \textbf{Elastic Model}:

$\bullet$ \textbf{Elastic Channel}: \textit{dynamically adjust the retrieval strategy}. A typical recommender system contains multiple retrieval channels. When CRs are insufficient, static strategies usually use fixed rules to drop some retrieval channels with high computation consumption. Different to static strategies, \textbf{\textit{Elastic Channel}} dynamically adjusts the retrieval strategy for each request according to the online environment and the features of the request.

$\bullet$ \textbf{Elastic Queue}: \textit{dynamically adjust the length of queue}. Under the limitation of CRs, recommender systems cannot provide the prediction service and ranking service for all candidate items. In static strategies, before entering the prediction service and ranking service, the queue of items needs to be truncated to a global fixed length. In contrast, \textbf{\textit{Elastic Queue}} dynamically adjusts truncation for each request length according to the online environment and the features of the request.

$\bullet$ \textbf{Elastic Model}: \textit{dynamically select prediction models}. 
Recommender systems often provide multiple prediction models with different computation consumption for one prediction service. A complex model achieves better revenue while taking more computation consumption. When CRs are insufficient, static strategies usually use fixed rules to downgrade high computation consumption models to low consumption models.
In contrast, \textbf{\textit{Elastic Model}} dynamically adjusts the prediction model for each request according to the online environment and the features of the request.

Recently, some dynamic strategies \cite{jiang2020dcaf,yang2021computation} have been proposed to achieve ``personalized" CR allocation. DCAF \cite{jiang2020dcaf} focuses on a single CR allocation phase. CRAS \cite{yang2021computation} focuses on multi-phase queue truncation problems, but it introduces some assumptions about Elastic Queue scenario. For example, it uses the queue length to represent the computation cost when modeling the CR allocation problem, and assumes that the revenue varies logarithmically with the queue length. However, these assumptions do not hold in Elastic Channel and Elastic Model scenarios. Moreover, existing studies ignore the state transition process of requests between different phases, which limits the effectiveness of their approaches. 


To address the limitations of existing studies, we propose RL-MPCA, which formulates the CR allocation problem as a Weakly Coupled Markov Decision Process (Weakly Coupled MDP) \cite{meuleau1998solving} problem and solves it with an RL-based approach. Compared to Constrained Markov Decision Process (CMDP) \cite{altman1999constrained}, Weakly Coupled MDP allows global weakly coupled constraints across sub-MDPs. Thus, it can model the problem of CR allocation across requests better than CMDP \cite{boutilier2016budget, chen2021primal, adelman2008relaxations}. 

Our main contributions are summarized as follows:
\begin{enumerate}
\item We propose an innovative CR allocation solution for recommender systems. To the best of our knowledge, this is the first work that formulates the CR allocation problem as a Weakly Coupled MDP problem and solves it with an RL-based approach.
\item We design a novel multi-scenario compatible Q-network adapting to the various CR allocation scenarios, then calibrate Q-value by introducing multiple adaptive Lagrange multipliers (adaptive-$\lambda$) to avoid violating the global CR constraints in training and serving.
\item We validate the effectiveness of our proposed RL-MPCA\footnote[3]{The publicly accessible code at \url{https://anonymous.4open.science/r/RL-MPCA-130D}.} approach through offline experiments and online A/B tests. Offline experiment results show that RL-MPCA can achieve better revenue than baseline approaches while satisfying the CR constraints.
Online A/B tests demonstrate the effectiveness of RL-MPCA in real-world industrial applications.

\end{enumerate}
\section{Related Work}
\subsection{CR Allocation and RL for Recommender Systems}
Recommender systems
 have been a popular topic in industry and academia in recent years. 
Most studies focus on improving the business revenue under the assumption of sufficient CRs \cite{zhao2020jointly, zhao2021dear}. 
Some of these studies focus on applying RL to recommender systems, including recommendations \cite{gao2019drcgr,zhou2020interactive,deng2021unified}, real-time bidding\cite{wu2018budget,tang2020optimized}, ad slots allocation \cite{xie2021hierarchical,zhao2021dear,liao2022cross}, etc.
Some studies concern CR consumption and try to reduce it through model compression \cite{polino2018model, fan2020training}. 
All the above studies rarely focus on CR allocation. 
As an exception, DCAF \cite{jiang2020dcaf} and CRAS \cite{yang2021computation} propose two ``personalized'' CR allocation approaches. They formulate the Elastic Queue CR allocation problem as an optimization problem, and then solve it with linear programming algorithms. 
%
Different from the above studies, our proposed RL-MPCA uses an RL-based dynamic CR allocation approach to improve the effectiveness.


\subsection{RL and Weakly Coupled MDPs}
A Weakly Coupled MDP \cite{meuleau1998solving} comprises multiple sub-MDPs, which are independent except that global resource constraints weakly couple them \cite{chen2021primal}. 
Due to the linking constraints, the scale of the problem grows exponentially in the number of sub-problems \cite{chen2021primal}.
Some studies try to relax Weakly Coupled MDP to CMDP \cite{altman1999constrained} and then solve it \cite{chen2021primal, adelman2008relaxations}.
The solutions to the CMDP problem include CPO \cite{achiam2017constrained}, RCPO \cite{tessler2018reward}, IPO \cite{liu2020ipo}, etc. They focus on the internal constraints of MDP. 
Recently, some studies focus on directly solving Weakly Coupled MDP problems. 
BCORLE($\lambda$) \cite{zhang2021bcorle} solves it with $\lambda$-generalization.
BCRLSP \cite{chen2022bcrlsp} first trains the unconstrained reinforcement model and then imposes a global constraint on the model with linear programming methods in near real-time. 
Both BCORLE and BCRLSP guarantee that budget allocations strictly satisfy a single global constraint.
CrossDQN \cite{liao2022cross} attempts to make the model avoid violating a single global constraint by introducing auxiliary batch-level loss. It uses a soft version of argmax to solve the problem of non-derivability of the native argmax function, which makes the model unable to strictly satisfy the global constraints during both offline training and online serving. 

Offline RL methods aim to learn effective policies from a fixed dataset without further interaction with the environment \cite{fujimoto2019off}. 
Off-policy methods (e.g., DQN \cite{mnih2015human}, DDQN \cite{van2016deep}) can be directly applied to Offline RL while ignoring the out-of-distribution (OOD) problem.
To solve the OOD problem, some offline RL methods are also proposed, including BCQ \cite{fujimoto2019benchmarking}, CQL \cite{kumar2020conservative}, COMBO \cite{yu2021combo}, etc. 
BCQ addresses the problem of extrapolation error via restricting the action space to force the agent towards behaving close to on-policy with respect to a subset of the given data. 
In addition, REM \cite{agarwal2020optimistic} enforces optimal Bellman consistency on random convex combinations of multiple Q-value estimates to enhance the generalization capability in the offline setting.
In the experiments of this paper, we choose three popular methods (DDQN, BCQ, and REM) as base models. Essentially, our proposed RL-MPCA only modifies the Q-network, so it can also apply to other Q-learning methods.

In addition, we can also consider the Weakly Coupled MDP problem as a black-box optimization problem, then solve it with evolutionary algorithms, such as Cross-Entropy Method (CEM) \cite{rubinstein2004cross} and Natural Evolution Strategies (NES) \cite{wierstra2014natural}.


\section{Problem Formulation}
\subsection{Original Problem Description}
The recent work \cite{jiang2020dcaf} formulated the single-phase CR allocation problem as a knapsack problem. Similarly, we formulate the multi-phase CR allocation problem as a knapsack problem.

\begin{align}
  \max_{j_1, \dots, j_T}\sum_{i=1}^M \sum_{j_1} \dots \sum_{j_T}\left( \prod_{t=1}^T x_{i,j_t} \right)Value_{i,j_i,\dots,j_T} \\
  s.t. \sum_{i=1}^M \sum_{j_1} \dots \sum_{j_T}\left(\prod_{t=1}^T x_{i,j_t} \right)Cost_{i,j_i,\dots,j_T} \leq C \label{formula:single-constraint} \\
  \sum_{j_t}x_{i,j_t} \leq 1, \ \ \forall i,t \\
  x_{i,j_t} \in \{0, 1\}, \ \ \forall i,j_t
\end{align}

We suppose there are $M$ online requests $\{i=1,\dots,M\}$ in a given time slice, and the maximum computation budget of the system in this time slice is $C$. For each request $i$, $T$ phases need to make computation decisions, and $N_t$ actions can be taken for the specified phase $t$. We define $j_1,\dots,j_T$ as a complete decision process of a request, and the decision action of phase $t$ is $j_t$ ($j_t \in \{1,\dots,N_t\}$). Meanwhile, for request $i$, if the decision process is $j_1,\dots,j_T$, we use $Value_{i,j_i,\dots,j_T}$ and $Cost_{i,j_i,\dots,j_T}$ to represent the expected revenue and computation cost, respectively. $x_{i,j_t}$ is the indicator that request $i$ is assigned action $j$ in phase $t$. In phase $t$, for request $i$, there is one and only one action $j_t$ can be taken.

InEq. (\ref{formula:single-constraint}) above assumes that all phases share an overall computation budget. However, in a real-world online recommender system, the CRs of each phase are often relatively independent. For example, recommender systems often deploy prediction and retrieval services on different clusters for ease of maintenance, and their CRs cannot be shared. 
Considering that each phase has a separate CRs budget, we replace the global constraint (InEq. (\ref{formula:single-constraint})) with multiple constraints InEq. (\ref{formula:multi-constraint}), where $Cost_{i,j_t}$ represents the computation cost when the decision of phase $t$ is $j_t$ for request $i$, and $C_t$ is the computation budget of phase $t$.
This paper focuses on the scenario of single-constraint CR allocation at each phase. If there is more than one constraint per phase, we can relax multiple constraints in the same phase and combine them into one.
\begin{align}
  s.t. \ \   \sum_{i=1}^M \sum_{j_t=1}^{N_t} x_{i,j_t} Cost_{i,j_t} \leq C_t, \ \ \forall t=1,\dots,T \label{formula:multi-constraint}
\end{align}

\subsection{Weakly Coupled MDP Problem Formulation \label{sub:weakly-coupled-mdp-problem-formulation}}
The decision results before phase $t$ affect the input state of phase $t$.
To better describe our approach, we take a three-phase CR allocation situation as an example in this paper. It contains one Elastic Channel phase, one Elastic Queue phase, and one Elastic Model phase, which is a typical case of recommender system CR allocation.
As shown in Figure \ref{fig:request-query-procedure}, for request $i$, the decision result of Elastic Channel phase determines the real retrieval queue, and it directly affects the input state of Elastic Queue phase. 
Similarly, the decision result of Elastic Queue phase affects the input state of Elastic Model phase. Therefore, to better adapt to the state transition process, in the multi-phase joint CR allocation, we introduce the ``state'' of the request.
\begin{figure}[htbp]
  \centering
  \includegraphics[width=0.9\linewidth]{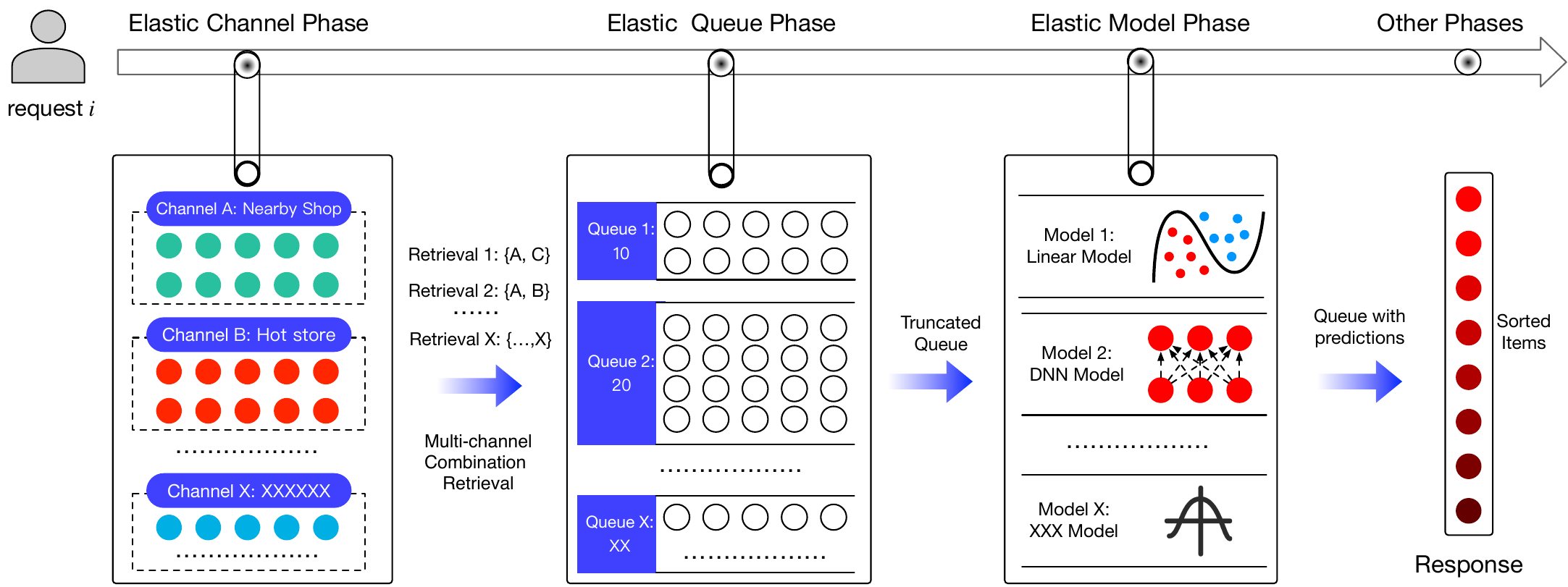}
  \caption{Request query procedure of recommender systems in a three-phase computation resource allocation situation.}
  \Description{}
  \label{fig:request-query-procedure}
\end{figure}


In this paper, we formulate the CR allocation problem as a Weakly Coupled MDP \cite{meuleau1998solving} problem. 
Formally, the Weakly Coupled MDP consists of a tuple of six elements ($\mathcal{S,A,R,P,\gamma,C}$), which are defined as follows:

\begin{itemize}
  \item \textbf{State Space $\mathcal{S}$}. For phase $t$ of request $i$, $s^i_t \in \mathcal{S}$ consists of user information $u$, time slice information $ts$, context information $c$, ad items information $\{ad_1, \dots, ad_{N_{ad}}\}$, and the CR allocation decision results of phase $t-1$.
  \item \textbf{Action Space $\mathcal{A}$}. Our CR allocation situation has three phases with different action spaces. The actions of the Elastic Channel phase, the Elastic Queue phase, and the Elastic Model phase are the retrieval strategy number, the truncation length, and the prediction model number, respectively.
  \item \textbf{Reward $\mathcal{R}$}. For request $i$, after the agent takes action for the final phase, the system returns the final sorted items to the user. The user browses the items and gives feedback, including order price $price_{o}$ and advertising fee $fee_{ad}$ of request $i$. The reward $r(s_t, a_t)$ is the weighted sum of them:
  \begin{align}
  r(s_t, a_t) = k_1 * fee_{ad} + k_2 * price_{o}
  \end{align}
  \item \textbf{Transition Probability $\mathcal{P}$}. $P(s_{t+1}|s_t, a_t)$ is the state transition probability from phase $t$ to phase $t+1$ after taking action $a_t$. 
  For each request $i$, trajectory ($\tau_i$) is its whole state transition process in the recommender system.
  \item \textbf{Discount Factor $\gamma$}. $\gamma \in [0,1]$ is the discount factor for future rewards. 
  \item \textbf{Global Constraint $\mathcal{C}$}. Each phase has its global constraint that couples sub-MDPs. InEq. (\ref{formula:multi-constraint}) defines these constraints. 
\end{itemize}

\section{Methodology}
\begin{figure}[htbp]
  \centering
  \includegraphics[width=0.9\linewidth]{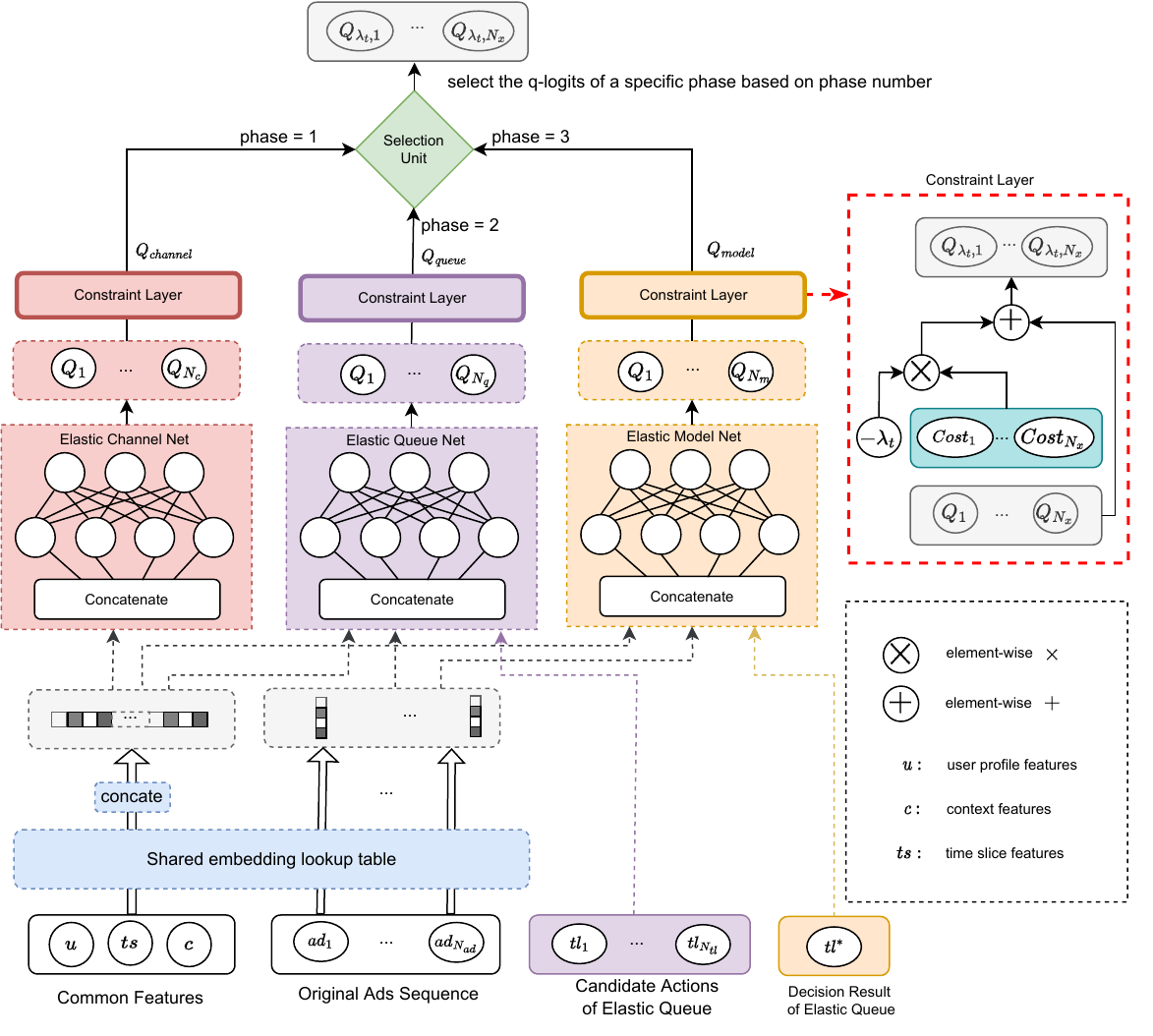}
  \caption{Q-Network of RL-MPCA. It first models each phase using a separate network and calibrates the Q-value with the constraint layer. Then the selection unit selects the q-logits of a specific phase based on the phase number $t$.}
  \Description{}
  \label{fig:Q-Network}
\end{figure}
Deep Q-Network (DQN) \cite{mnih2015human} and its improved versions \cite{wang2016dueling, van2016deep, fujimoto2019benchmarking} are very popular in solving sub-MDPs with discrete actions. The Q-network is the essential structure of these models, $Q_{\theta}(s, a)$. To adapt to various CR allocation scenarios, we design a novel deep Q-Network with multiple separate networks (As Figure \ref{fig:Q-Network} shows). In particular, the state space of each phase is defined as follows:

  $\bullet$  In the Elastic Channel phase, the candidate action space is the retrieval strategy numbers.
  For a recommender system with $N_r$ retrieval channels, the number of retrieval strategies is $N_c = 2^{N_{r}}$ and the candidate action space is $\{1,\dots,N_c\}$. 
  For example, for three candidate retrieval channels $\{A, B, C\}$, retrieval strategy $(0,1,1)$ indicates that channel A is not retrieved, and channels B and C are retrieved. We convert the indicator vector as a binary value, then the strategy number of the strategy $(0,1,1)$ is the integer $3$.

  $\bullet$  In the Elastic Queue phase, the action space is the truncation length. To reduce the candidate action space, we can put the candidate actions into buckets, e.g., set every ten adjacent truncation lengths as one bucket. Then the candidate action space is $\{10, 20,\dots\}$.
  
  $\bullet$  In the Elastic Model phase, the candidate action space is the prediction model numbers $\{1,\dots,N_m\}$.

Each phase of CR allocation has its own action spaces. As Figure \ref{fig:Q-Network} shows, we model each phase using separate networks to adapt the different action spaces. In the last layer of the Q-Network, we use the selection unit to select the q-logits of a specific phase based on phase number $t$.

\subsection{Constraint Layer \label{sec:constraint-layer}}
For any phase $t$, suppose that we have the optimal policy $\pi_{\neg t}^{*}$ that satisfies the constraints of all phases except $t$.
Then the decision problem for current phase $t$ can be modeled separately as the following single-phase CR allocation problem with a single constraint:
\begin{align}
  \max_{a_t}\sum_{i=1}^M \sum_{a_t=1}^{N_t} x_{i,a_t} {Value}_{i,a_t} \\
  s.t.  \sum_{i=1}^M \sum_{a_t=1}^{N_t} x_{i,a_t} {Cost}_{i,a_t} \leq C_t \\
   \sum_{a_t=1}^{N_t}x_{i,a_t} \leq 1, \ \ \forall i \\
   x_{i,a_t} \in \{0, 1\}, \ \ \forall i,a_t
\end{align}

By constructing and solving the Lagrange dual problem, we have the optimal solution to this problem. The proof is provided in Appendix \ref{sec:proof}.
For request $i$, the optimal action of phase $t$ is $a_t^*$:
\begin{align}
  a_t^* = \mathop{\arg \max}_{a_t}({Value}_{i,a_t} - \lambda_t Cost_{i, a_t}) \label{formula:optimal-action-for-value-cost}
\end{align}
where $\lambda_t \geq 0$ is the Lagrange multiplier.

Further, we use $Q^{\pi_{\neg t}^{*}}(s_t,a_t)$ to represent the expected cumulative reward for taking action $a_t$ in state $s_t$ and subsequent actions are decided following policy $\pi_{\neg t}^{*}$.
\begin{align}
  Q^{\pi_{\neg t}^{*}}(s_t,a_t) = \mathbb{E}_{\tau \sim \pi_{\neg t}^{*}} [R_t|s_t,a_t] \\
  R_t = \sum_{i=t+1}^{\infty} \gamma^{i} r(s_i, a_i, s_{i+1})
\end{align}

For phase $t$ of request $i$, we have:
\begin{align}
 Q^{\pi_{\neg t}^{*}}(s_t,a_t) = {Value}_{i,a_t} \\
 Cost(s_t,a_t) = {Cost}_{i,a_t}
\end{align}
where $Cost(s_t,a_t)$ is the computation cost for taking $a_t$ in $s_t$, determined by $(s_t, a_t)$, and independent of both prior and subsequent strategies. Thus, for phase $t$, the optimal action for request $i$ in state $s_t$ is $a_t^*$:
\begin{align}
  a_t^* = \mathop{\arg \max}_{a_t}(Q^{\pi_{\neg t}^{*}}(s_t,a_t) - \lambda_t Cost(s_t, a_t))
\end{align}

Compared to the action selection formula in original DQN \cite{mnih2015human} networks, we only need to add a layer (Constraint Layer) to obtain the optimal action that satisfies the CR constraints.
\begin{align}
Q_{\lambda_t}^{\pi_{\neg t}^{*}}(s_t,a_t) = Q^{\pi_{\neg t}^{*}}(s_t,a_t) - \lambda_t Cost(s_t, a_t)
\end{align}

Then the optimal action is:
\begin{align}
  a_t^* =\mathop{\arg \max}_{a_t} Q_{\lambda_t}^{\pi_{\neg t}^{*}}(s_t,a_t)
\end{align}

\subsubsection{Adaptive-$\lambda$ in Offline Model Training}
As mentioned in DCAF \cite{jiang2020dcaf} and CRAS \cite{yang2021computation}, Assumptions (\ref{assumption:1}) and (\ref{assumption:2}) usually hold in general recommender systems. 
\begin{assumption} \label{assumption:1}
  $Value_{i,a_t}$ is monotonically increasing with $Cost_{i,a_t}$.
\end{assumption}
\begin{assumption} \label{assumption:2}
  $\frac{Value_{i,a_t}}{Cost_{i,a_t}}$ is monotonically decreasing with $Cost_{i,a_t}$.
\end{assumption}
From our observations, they also hold for most requests in Meituan advertising system. However, it is worth noting that our assumptions differ from those of CRAS. CRAS uses the queue length to represent the computation cost, while we make no assumptions about the relationship between computation cost and queue length (or other actions).


Given a fixed $\{Value_{i,a_t}\}_{i=1}^M$ and variable $\lambda_t$, the optimal action of phase $t$ is $a_t^*$ (Eq. \ref{formula:optimal-action-for-value-cost}). 
For each request $i$, its optimal action $a_t^*$ varies with $\lambda_t$, then the total computation cost $\hat{C}_t(\lambda_t)$ (Eq. \ref{formula:total-cost-for-lambda}) and the total revenue  $\hat{R}_t(\lambda_t)$ (Eq. \ref{formula:total-revenue-for-lambda}) of phase $t$ vary with $\lambda_t$.
\begin{align}
  \hat{C}_t(\lambda_t) &= \sum_{i=1}^M {Cost}_{i,a_t^*} \label{formula:total-cost-for-lambda} \\
  \hat{R}_t(\lambda_t) &= \sum_{i=1}^M {Value}_{i,a_t^*} \label{formula:total-revenue-for-lambda}
\end{align}


We can obtain the optimal $\lambda_t$ which satisfies the CR constraint and maximizes $\hat{R}_t(\lambda_t)$ through updating $\lambda_t$ iteratively based on $\hat{C}_t(\lambda_t)$.
\begin{lemma} \label{lemma:1}
  Suppose Assumptions (\ref{assumption:1}) and (\ref{assumption:2}) hold, for any $\lambda_t^k$, let $\lambda_t^{k+1}$ be:  
  \begin{align}
    \lambda_t^{k+1} \leftarrow \lambda_t^{k} + \alpha \left( \frac{\hat{C}_t(\lambda_t^k)}{C_t} - 1 \right) \label{formula:update-lambda}
  \end{align}
  where $C_t$ is the computation budget of phase $t$ and $\alpha \in \mathbb{R^{+}}$ is learning rate of $\lambda$. Then, the following conclusion holds: 
  \begin{itemize}
    \item Conclusion 1. $\hat{C}_t(\lambda_t^{k+1}) \leq \hat{C}_t(\lambda_t^{k})$ will holds if $\hat{C}_t(\lambda_t^{k}) > C_t$.  \label{lemma:1.1}
    \item Conclusion 2. $\hat{R}_t(\lambda_t^{k+1}) \geq \hat{R}_t(\lambda_t^{k})$ will holds if $\hat{C}_t(\lambda_t^{k}) < C_t$.  \label{lemma:1.2}
    \item Conclusion 3. $\lambda_t^{k+1} = \lambda_t^{k}$ will holds if $\hat{C}_t(\lambda_t^{k}) = C_t$. \label{lemma:1.3}
  \end{itemize}
\end{lemma}

\begin{proof}
  Suppose Assumptions (\ref{assumption:1}) and (\ref{assumption:2}) hold, $\hat{C}_t(\lambda_t)$ is monotonically decreasing with $\lambda_t$ (see more details in \cite{jiang2020dcaf}). Further, under Assumption \ref{assumption:2}, $\frac{\hat{R}_t(\lambda_t)}{\hat{C}_t(\lambda_t)}$ is monotonically decreasing with $\lambda_t$, and under Assumption \ref{assumption:1}, $\hat{R}_t(\lambda_t)$ is monotonically decreasing with $\lambda_t$. 
  \begin{itemize}
    \item when $\hat{C}_t(\lambda_t^{k}) > C_t$, we have $\lambda_t^{k+1} > \lambda_t^{k}$, then $\hat{C}_t(\lambda_t^{k+1}) \leq \hat{C}_t(\lambda_t^{k})$ holds.
    \item when $\hat{C}_t(\lambda_t^{k}) < C_t$, we have $\lambda_t^{k+1} < \lambda_t^{k}$, then $\hat{R}_t(\lambda_t^{k+1}) \geq \hat{R}_t(\lambda_t^{k})$ holds.
    \item when $\hat{C}_t(\lambda_t^{k}) = C_t$, we have $\lambda_t^{k+1} = \lambda_t^{k}$ holds.
  \end{itemize}
\end{proof}

In summary, Lemma (\ref{lemma:1}) specifies that it is feasible to update $\lambda$ with formula (\ref{formula:update-lambda}). 
Initially, conclusion 1 of Lemma (\ref{lemma:1}) indicates that when the total computation cost exceeds the computation budget, updating $\lambda_t$ with formula (\ref{formula:update-lambda}) will obtain less total computation cost. It helps to avoid violating the constraint. 
Furthermore, conclusion 2 of Lemma (\ref{lemma:1}) indicates that when the total computation cost is less than computation budget, updating $\lambda_t$ with formula (\ref{formula:update-lambda}) will obtain a better total revenue. 
Finally, conclusion 3 of Lemma (\ref{lemma:1}) indicates that when the total computation cost equals to computation budget, updating $\lambda_t$ with formula (\ref{formula:update-lambda}) will obtain the original value of $\lambda_t$. 
Updating $\lambda_t$ with formula (\ref{formula:update-lambda}) until convergence, we will obtain the optimal $\lambda_t^*$, where $\hat{C}_t(\lambda_t^{*}) = C_t$.

\begin{algorithm}
	\renewcommand{\algorithmicrequire}{\textbf{Input:}}
	\renewcommand{\algorithmicensure}{\textbf{Output:}}
	\caption{Offline Training of RL-MPCA (Based on DDQN)}
	\label{alg:offline-training}
	\begin{algorithmic}[1]
		\REQUIRE Dataset $\mathcal{D}$, number of iteration $I$, mini-batch size $N$, adaptive-$\lambda$ update times $K$
    \STATE Initialize Q-network $Q_{\theta}$, target Q net $Q_{{\theta}^{'}}$ ($\theta^{'} \leftarrow \theta$), Lagrange multipliers $\boldsymbol{\lambda}=(\lambda_1,\dots,\lambda_T)$
		\FOR{$i = 1,\dots,I $}
		\STATE Sample batch $\mathcal{D}^{i}$ of $N$ transitions $(s_t, a_t, r_t, s_{t+1})$ from $\mathcal{D}$
    \STATE $a_{t+1} = \mathop{\arg \max}_{a_{t+1}} \left(Q_{\theta}(s_{t+1}, a_{t+1})-\lambda_{t+1}^{i} Cost(s_{t+1}, a_{t+1})\right)$
    \STATE $\theta \leftarrow \mathop{\arg \min}_{\theta} \sum_{\mathcal{D}^{i}} \left(r_t + \gamma Q_{\theta^{'}}(s_{t+1}, a_{t+1})-Q_{\theta}(s_t, a_t)\right)^2 $ 
      \FOR{$k = 1,\dots,K$}
      \STATE For $s_t \in \mathcal{D}^{i}$, take $a_t^k$ with (\ref{forumla:adaptive-lambda-action-taking})
      \STATE For $t \in \{1,\dots,T\}$, update $\lambda_t^{i,k+1}$ with (\ref{forumla:adaptive-lambda-update})
      \ENDFOR
    \STATE $\boldsymbol{\lambda}^{i+1} \leftarrow \boldsymbol{\lambda}^{i}$
		\STATE Every $N_{target}$ steps reset $\theta^{'} \leftarrow \theta$
		\ENDFOR
		\ENSURE $Q_{\theta}$, $\boldsymbol{\lambda}=(\lambda_1,\dots,\lambda_T)$
	\end{algorithmic}
\end{algorithm}

As described in Algorithm \ref{alg:offline-training}, we dynamically update the $\lambda$ in the offline training phase. At iteration step $i$, we take a mini-batch of samples $\mathcal{D}_i$ (a bigger batch is generally taken here, e.g., 8192 samples per batch), and update $\boldsymbol{\lambda} = (\lambda_1,\dots,\lambda_T)$ $K$ times. 
At the $k$-th update, for each $s_t$ in $\mathcal{D}_i$, take action $a_t^k$ with: 
\begin{align}
  a_t^k = \mathop{\arg \max}_{a_{t}} \left(Q_{\theta}(s_{t}, a_{t})-\lambda_{t}^{i,k} Cost(s_{t}, a_{t})\right) \label{forumla:adaptive-lambda-action-taking}
\end{align}

and for each phase $t \in \{1,\dots,T\}$ at the $k$-th update, update $\lambda_t^{i,k+1}$ with:
\begin{align}
  \lambda_t^{i,k+1} \leftarrow \max \left\{0, \lambda_t^{i,k} + \alpha \left(\frac{\sum_{(s_t,a_t^k) \in \mathcal{D}^{i,k}}{Cost}(s_t,a_t^k)}{{C}_t({\mathcal{D}^{i}})} - 1\right) \label{forumla:adaptive-lambda-update}\right\}
\end{align}

where $\alpha \in \mathbb{R^{+}}$ is the learning rate of adaptive-$\lambda$, and $C_t(\mathcal{D}_i)$ is the maximum CR budget that the system can allocate for dataset $\mathcal{D}_i$ at phase $t$. 
$C_t(\mathcal{D}_i)$ can be calculated through an offline fixed rule, which is designed by stress testing and practical experience.


Algorithm \ref{alg:offline-training} describes the training process of DDQN-based RL-MPCA. Essentially, the Constraint Layer module of RL-MPCA only modifies the Q-network, so it can also apply to other Q-learning methods. Take a popular offline RL method with Q-network, REM \cite{agarwal2020optimistic}, for example. The only difference between Algorithm \ref{alg:offline-training} and the REM-based RL-MPCA approach is Q-network $Q_{\theta}$. Specifically, for REM model with $H$ heads, we replace $Q_{\theta}$ with 
$Q_{\theta}^{REM} = \sum_{h} \beta_h Q_{\theta}^h$,
where $\boldsymbol{\beta} = (\beta_1,\dots,\beta_H)$ is categorical distribution, which is randomly drawed for each mini-batch (see more details in \cite{agarwal2020optimistic}).

\subsubsection{$\lambda$ Correction in Offline Model Evaluation \label{sec:lambda-correction}}
After the offline model is trained, the $\boldsymbol{\lambda}$-calibrated Q-value guarantees that the agent's decisions satisfy the CR constraints on all training datasets. 
However, when applying $\boldsymbol{\lambda}$ to the real online system, it still faces the following problems: 
(1) The online and offline data distributions are inconsistent because the behavioral policy of collecting offline data differs from the target policy, which leads to the possibility that $\boldsymbol{\lambda}$ may not satisfy the CR constraints in the online system. 
(2) The traffic of the recommender system varies over time, and the existing $\boldsymbol{\lambda}$ cannot satisfy the CR constraints on each time slice.

To solve problem (1), we build an offline simulation system, which interacts with the agent and gives feedback on the computation cost and revenue in imitation of the real online environment. Through the evaluation in the simulation system, we select the optimal $\boldsymbol{\lambda}^*$ that satisfies the CR constraints in order of the decision phases. When both Assumptions (\ref{assumption:1}) and (\ref{assumption:2}) hold, we can find optimal $\boldsymbol{\lambda}^*$ through bisection search in each phase (please refer to \cite{jiang2020dcaf} for the detailed proof). Otherwise, we can find optimal $\boldsymbol{\lambda}^*$ through grid search \cite{bergstra2012random}.

To solve problem (2), we select the optimal $\boldsymbol{\lambda}^*$ for each time slice based on the offline simulation system. We observe that the traffic of the recommender system generally varies periodically except for special holidays. Taking Meituan advertising system as an example, its traffic variation cycle is one day. 
Therefore, we can divide a day into multiple time slices with similar traffic distribution in the same time slice. Considering that the cost of training a separate model for each time slice is expensive and not easy to maintain, we first train a uniform model for all time slices, and then solve a separate $\boldsymbol{\lambda}$ for each time slice.
Alternatively, for systems with non-periodic traffic, a possible solution is to use the traffic from the previous time slice to represent the current time slice. Specifically, we can update $\boldsymbol{\lambda}$ in near real-time, thus allowing $\boldsymbol{\lambda}$ to automatically adapt to irregular traffic changes.
\subsection{System Architecture}
We illustrate the overview of the architecture in Figure \ref{fig:system-architecture}. In each phase, for instance, in the Elastic Queue phase, we need to allocate and control CRs through Computation Allocation System and Computation Control System.
Computation Allocation System aims to maximize the total business revenue under the CR constraints.
Computation Control System aims to guarantee system stability by means of feedback control. Dynamic allocation of CRs poses a significant challenge in guaranteeing the stability of recommender systems. We use Flink \cite{carbone2015apache} to collect real-time system load information, such as failure rate, CPU utilization, etc., and then use PID \cite{ang2005pid} control algorithms to achieve feedback control. When the system load exceeds the target value of the PID, the PID will control the consumption of CRs. For instance, in the Elastic Queue scenario, when the system's failure rate rises above the target value, the PID Controller will reduce the upper bound of the queue for all requests. The result of online A/B tests shows that the Computation Control System reduces the degradation rate by 0.1 percentage point, provides automatic and timely responses to unexpected traffic, and guarantees the stability of recommender systems.
\begin{figure}[htbp]
  \centering
  \includegraphics[width=0.95\linewidth]{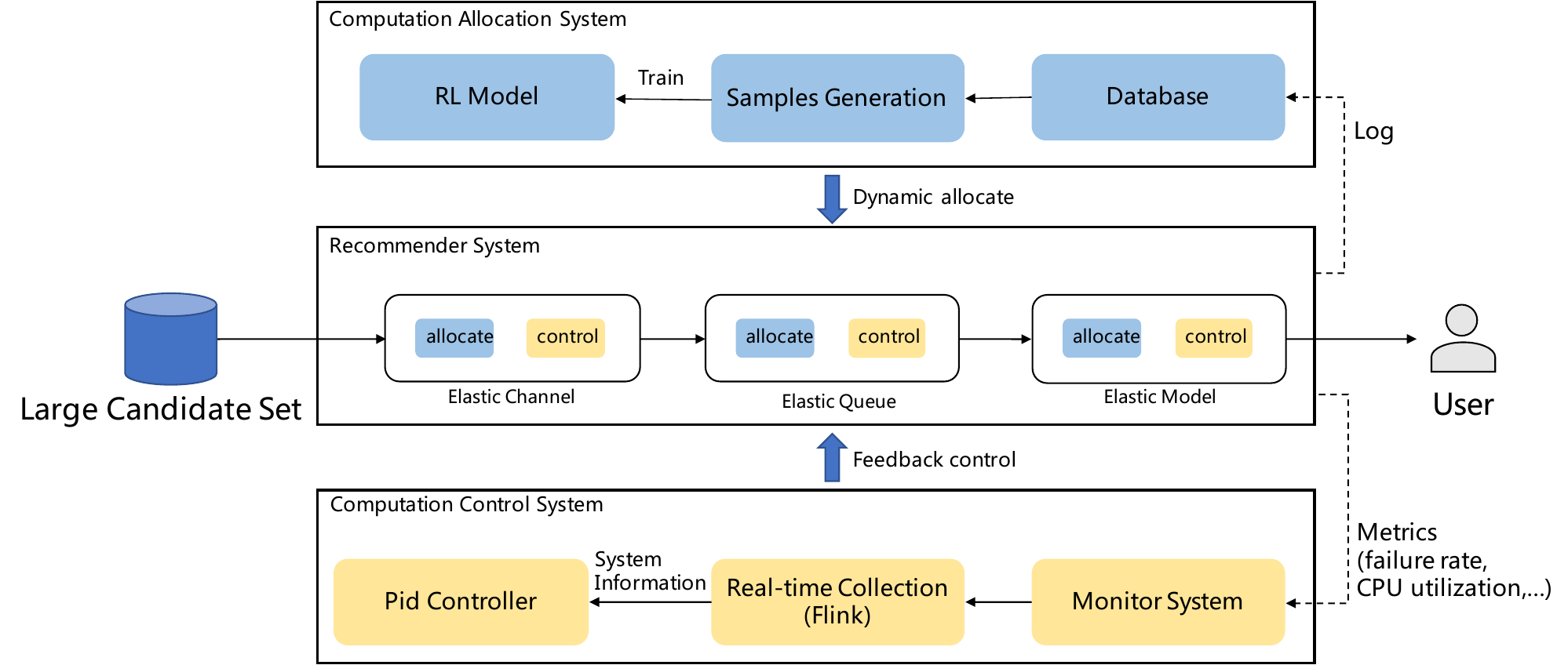}
  \caption{The Overview of System Architecture.}
  \Description{}
  \label{fig:system-architecture}
  \vspace{-3mm}
\end{figure}

\section{Experiments}
Our experiments aim to study four questions: 
(1) Does adaptive-$\lambda$ of constraint layer help to avoid violating the global CR constraints in the training process?
(2) After $\lambda$ correction, does the model with constraint layer satisfy the global CR constraints and improve business revenue?
(3) How does RL-MPCA approach perform in comparison to other state-of-the-art CR allocation approaches and RL algorithms?
(4) How do different hyper-parameter settings affect the performance of RL-MPCA?

To answer these questions, we conduct various experiments in a three-phase joint modeling CR allocation situation, which contains one Elastic Channel phase, one Elastic Queue phase, and one Elastic Model phase. 

\subsection{Offline Experiments}
To demonstrate the performance of the proposed RL-MPCA, we evaluate and compare various related approaches for CR allocation on a real-world dataset. In offline experiments, we use the simulation system to evaluate these approaches.

\subsubsection{Dataset}

We run random exploratory policies and superior policies (see more details about behavioral policies in Appendix \ref{sec:behavioral-policies}) to collect the dataset on Meituan advertising system during July and August 2022. Finally, we sample 568,842,204 requests from 101,368,290 users as the dataset, which includes user profile features, context features, time slice features, etc.
\subsubsection{Offline Simulation System \label{sec: offline-simulation-system}}
It is dangerous to deploy a model to an online system when its effect is unknown, which may significantly damage the online revenue of the recommender system and cause the online service to crash. To solve this problem, we build an offline simulation system, which can imitate the online real-world environment to interact with the model (agent) and give feedback on the computation consumption and revenue. 
More details about the offline simulation system are described in Appendix \ref{sec:simulation-system-appendix}.

\subsubsection{Evaluation Metrics}
We use computation ($cost$) and revenue ($return$) to evaluate the performance of approaches in offline experiments. The computation cost of each phase is defined as the sum of the CR consumption of all requests in that phase (see more details in Appendix \ref{sec:computation-cost-estimation}). To facilitate analysis, we define total computation cost as ($cost = \sum_t (\frac{\hat{C}_t}{C_t}-1)$).
$return$ is defined as the total revenue of all requests, specifically, $return = \sum fee_{ad} + \sum price_{o}$. 
With reference to D4RL \cite{fu2020d4rl}, to facilitate the analysis of the effectiveness of different approaches while ignoring the impact of our application scenarios, we normalize scores by: 
\begin{align}
  normalized\_score = 100 * \frac{score-random\_score}{expert\_score - random\_score} \label{formula:normalize_score}
\end{align}
\subsubsection{Hyper-parameters Settings}
RL-MPCA contains several hyper-parameters. We employed the grid search \cite{bergstra2012random} to determine the hyper-parameter values. Appendix \ref{sec:hyper-parameters} provides the hyper-parameters of experiments.

\subsubsection{Baselines}
We compare RL-MPCA with several baselines. 
Our situation has only one Elastic Queue phase (the two other phases are Elastic Channel and Elastic Model). In the situation containing only one Elastic Queue phase, the modeling methods of DCAF and CRAS are consistent. Therefore, in the later experiments, we only show the details of DCAF.
\begin{itemize}
  \item \textbf{Static}. Static approach allocates CRs with global fixed rules, including fixed retrieval channels, fixed truncation length of candidate items, and fixed prediction models.
  \item \textbf{DCAF}. DCAF \cite{jiang2020dcaf} formulates the CR allocation problem as an optimization problem with constraints, then solves the optimization problem with linear programming algorithms. In online A/B tests, we use fixed rules in Elastic Channel phase and Elastic Model phase, and DCAF is deployed in the Elastic Queue phase.
  \item \textbf{ES-MPCA}. Before RL-MPCA, we designed an evolutionary strategies based multi-phase computation allocation approach (ES-MPCA, see more details in Appendix \ref{sec:es-mpca}), which has been deployed on Meituan advertising system. 
  \item \textbf{Ex-RCPO}. RCPO \cite{tessler2018reward} solves a CMDP problem by introducing the penalized reward functions (i.e., calibrate rewards with Lagrange multiplier $\lambda$). We replace adaptive-$\lambda$ of RL-MPCA with the penalized reward functions when training the model, and name it Ex-RCPO.
  \item \textbf{Ex-BCORLE($\lambda$)}. BCORLE \cite{zhang2021bcorle} solves a single-constraint budget allocation problem with $\lambda$-generalization. It cannot be directly applied to the multi-constraint CR allocation. We extend BCORLE from single-$\lambda$ to multi-$\lambda$, and name it Ex-BCORLE.
  \item \textbf{Ex-BCRLSP}. BCRLSP \cite{chen2022bcrlsp} solves the single-constraint budget allocation problem by calibrating Q-value in near real-time. It cannot be directly applied to the multi-constraint CR allocation. We extend BCORLE from single-$\lambda$ to multi-$\lambda$, and name it Ex-BCORLE.
  \item \textbf{Ex-CrossDQN}. CrossDQN \cite{liao2022cross} solves a single-constraint ads allocation problem by introducing auxiliary batch-level loss when training the model. We replace adaptive-$\lambda$ of RL-MPCA with auxiliary batch-level loss when training the model, and name it Ex-CrossDQN.
\end{itemize}


\subsubsection{Offline Experiment Results} 

\begin{figure}[htbp]
  \centering
  \captionsetup[subfigure]{labelformat=empty} 
  \setlength{\abovecaptionskip}{0.cm}
  \subfigure[overall cost and return of the three phases] {
    \includegraphics[width=\linewidth, trim=100 50 100 50, clip]{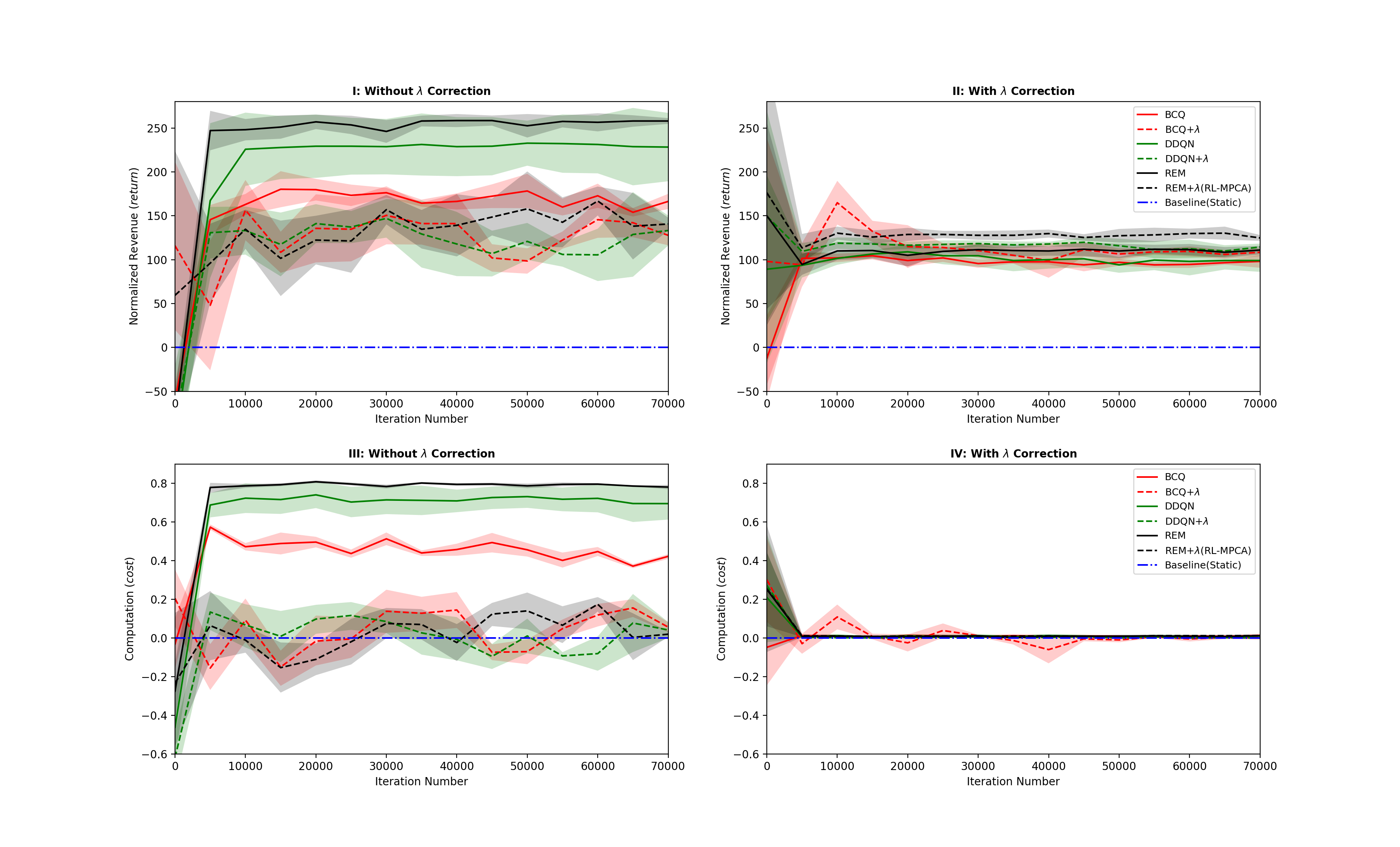}
  }

  \subfigure[cost at each phase without $\lambda$ correction] {
    \includegraphics[width=\linewidth, trim=100 0 100 25, clip]{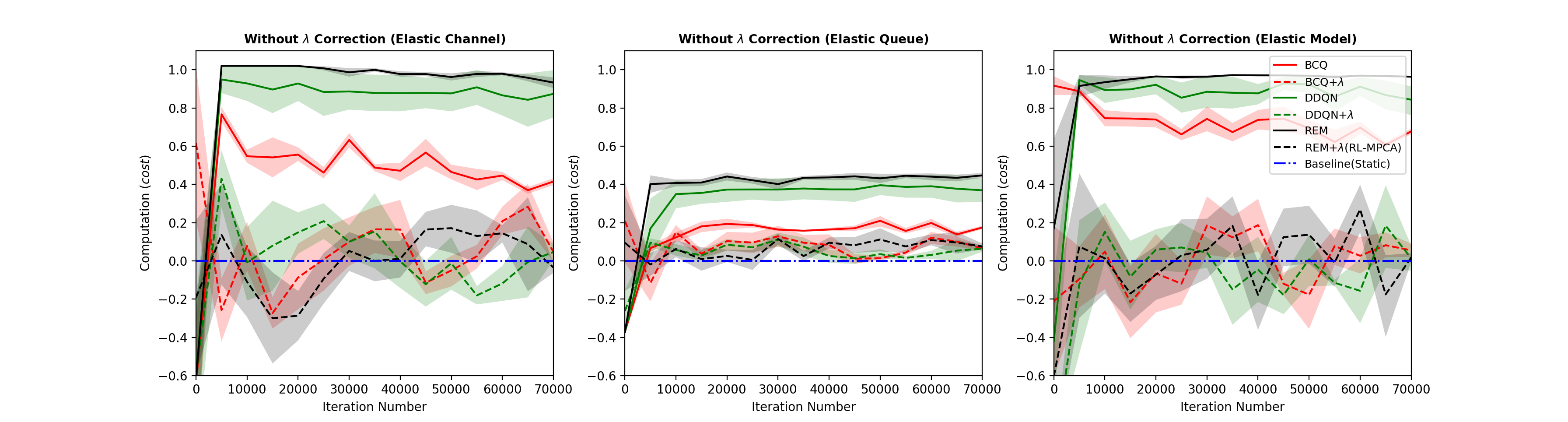}
  }
  \caption{Offline experiment results for adaptive-$\lambda$ and $\lambda$ correction on multiple Deep Q-Network models. Agents are evaluated every 5,000 steps, and averaged over 5 seeds.}
  \Description{}
  \label{fig:offline-exp-results-of-multiple-models-day}
\end{figure}

To answer question (1) and question (2), we train multiple models: DDQN, BCQ, REM, and their improved versions of introducing adaptive-$\lambda$, then use the simulation system to evaluate them.
As shown in Figure \ref{fig:offline-exp-results-of-multiple-models-day}, during the training process, introducing adaptive-$\lambda$ can control the CRs of the model always around the target constraints for each phase (Figure \ref{fig:offline-exp-results-of-multiple-models-day}.a.\uppercase\expandafter{\romannumeral3} and Figure \ref{fig:offline-exp-results-of-multiple-models-day}.b). 
However, the CRs of models swing around the target constraints due to the inconsistent distribution of the mini-batch sampled during training and the evaluation dataset (see more details in Section \ref{sec:lambda-correction}). 
After the $\lambda$ correction, for each phase, the CRs of models strictly satisfy the constraints except for BCQ+$\lambda$ (BCQ with adaptive-$\lambda$), and Figure \ref{fig:offline-exp-results-of-multiple-models-day}.a.\uppercase\expandafter{\romannumeral4} shows the total CRs of all phases.
Adaptive-$\lambda$ allows the model to learn the Q-value under the case that CRs conform to the constraint (or in the near range of the constraints) at each phase. Thus, we can observe that the effectiveness of all three models improves after introducing adaptive-$\lambda$ (Figure \ref{fig:offline-exp-results-of-multiple-models-day}.a.\uppercase\expandafter{\romannumeral2}). 
An interesting phenomenon is that after introducing adaptive-$\lambda$, the CRs of BCQ instead cannot be stably calibrated to conform to the constraint. 
A potential reason is that the BCQ model contains an imitation component, which causes the BCQ model to imitate the behavioral strategy. As a result, the value of $\lambda$ changes in an unknown direction during the training process, and eventually, the Q value cannot be calibrated to satisfy the target constraints.

\begin{table*}[htbp]
  \centering
  \tabcolsep=0.3cm
  \begin{tabular}{lccccc}\toprule
     & \multicolumn{2}{c}{Before Calibration} & \multicolumn{3}{c}{After Calibration}
     \\\cmidrule(lr){2-3}\cmidrule(lr){4-6}
              & $cost$  & $return$   & ConstraintSat & $cost$ & $return$ \\
     \midrule
     Static                  & 100\% & - & Yes & 100\% & - \\
     DCAF                   & - & - & Yes & $100(\pm0.5)\%$ & $52.7(\pm0.4)$ \\
     ES-MPCA                & - & - & Yes & $100(\pm0.5)\%$ & $77.8(\pm0.1)$ \\
     Ex-RCPO                & $109.7(\pm20.0)\%$ & $146.9(\pm65.4)$ & Yes & $100(\pm0.5)\%$ & $111.8(\pm8.4)$ \\
     Ex-BCORLE($\lambda$)   & - & - & Yes & $100(\pm0.5)\%$ & $74.2(\pm36.1)$ \\
     Ex-CrossDQN            & $97.4(\pm5.6)\%$ & $98.4(\pm32.8)$ & Yes & $100(\pm0.5)\%$ & $112.0(\pm10.3)$ \\
     DDQN                   & $172.6(\pm12.5)\%$ &  $227.6(\pm16.2)$ & Yes & $100(\pm0.5)\%$ & $100.0(\pm12.9)$ \\
     BCQ                    & $146.4(\pm7.6)\%$ &  $169.1(\pm16.2)$ & Yes & $100(\pm0.5)\%$ & $97.1(\pm5.5)$ \\
     REM(Ex-BCRLSP)         & $179.2(\pm2.0)\%$ &  $254.1(\pm11.8)$ & Yes & $100(\pm0.5)\%$ & $108.9(\pm9.0)$ \\
     DDQN+$\lambda$         & $102.7(\pm17.8)\%$ &  $125.3(\pm31.8)$ & Yes & $100(\pm0.5)\%$ & $115.4(\pm7.4)$ \\
     BCQ+$\lambda$          & $101.5(\pm17.4)\%$ &  $119.3(\pm42.8)$ & No/Yes & -/$100(\pm0.5)\%$ & -/$108.7(\pm12.4)$ \\ 
     \textbf{REM+$\lambda$(RL-MPCA)} & $103.4(\pm18.3)\%$ &  $135.4(\pm37.0)$ & Yes & $100(\pm0.5)\%$ & $\boldsymbol{126.2}(\pm8.7)$
     \\\bottomrule
  \end{tabular}
  \caption{The offline results in the simulation system. All numbers of $return$ are the normalized score calculated by Eq. (\ref{formula:normalize_score}), where $random\_score$ and $expert\_score$ are the scores of Static and DDQN.}
  \label{tab:offline-experiments-results}
  \vspace{-3mm}
 \end{table*}

To answer question (3), we compare RL-MPCA to the state-of-the-art CR allocation approach DCAF and other related approaches. The results are shown in Table \ref{tab:offline-experiments-results}. Experiment results show that RL-MPCA outperforms other approaches in $return$ when the CR constraints are satisfied.

\subsubsection{Hyper-parameter Analysis}
To answer question (4), We compare the effect of two critical parameters, $\alpha$ and $K$, on the performance of RL-MPCA. $\alpha$ is the learning rate of adaptive-$\lambda$. $K$ is $\lambda$ update times in one global step.

\textbf{Hyper-parameter $\alpha$}. Like the learnng rate of the model's common parameters, the learning rate of adaptive-$\lambda$ $\alpha$ cannot be too big or too small. Too small a learning rate will lead to slow learning, while too big a learning rate will cause $\lambda$ to swing around the optimal value. Table \ref{tab:hyper-parameters} shows the model performance at different learning rates, and we finally choose $0.1$ as the parameter value.

\textbf{Hyper-parameter $K$}. A bigger $K$ indicates more updates to the $\lambda$ at once update of the model parameter during training, which will make constraints easier to be satisfied. As seen in Table \ref{tab:hyper-parameters}, the $return$ increases with the increase of $K$. However, a bigger $K$ also means more time consumption for training. To trade off the revenue and time, we choose $10$ as the parameter value.

 \begin{table}[htbp]
  \centering
  \begin{tabular}{cc|ccc}\toprule
    \multicolumn{2}{c}{$\alpha$} & \multicolumn{3}{c}{$K$}
     \\\cmidrule(lr){1-2}\cmidrule(lr){3-5}
     $\alpha$ & $return$ & $K$ & $return$ & training-time \\
     \midrule
      0.001   & $117.2(\pm6.5)$ & 1 & $118.8(\pm14.1)$ & 100\% \\
      0.01    & $122.9(\pm13.2)$ & 5 & $120.5(\pm10.4)$ & 115\% \\
      0.05    & $123.8(\pm8.1)$ & \textbf{10} & $\boldsymbol{126.2}(\pm8.7)$ & 153\% \\
      \textbf{0.1} & $\boldsymbol{126.2}(\pm8.7)$ & 15 & $\boldsymbol{126.3}(\pm10.7)$ & 200\% \\
      0.5     & $122.1(\pm10.2)$ & 20 & $\boldsymbol{125.4}(\pm10.2)$ & 231\% \\
      1.0     & $113.0(\pm7.1)$  & 30 & $\boldsymbol{127.1}(\pm11.1)$ & 253\% \\
     \bottomrule
  \end{tabular}
  \caption{Experiment results of hyper-parameters $\alpha$ and $K$.}
  \label{tab:hyper-parameters}
  \vspace{-3mm}
 \end{table}

\subsection{Online A/B test Results}
We also evaluate the RL-MPCA approach for two weeks in the online environment. 
In online A/B tests, we compare our proposed RL-MPCA approach with several previous strategies deployed on Meituan advertising system. 
Table \ref{tab:online-experiments-results} lists the performance of several primary online metrics, including gross merchandise volume per mille (GPM, i.e., $GPM = \text{avg}(price_{o}) * 1000$, where $\text{avg}(price_{o})$ is the average of $price_{o}$), cost per mille (CPM, i.e., $CPM = \text{avg}(fee_{ad}) * 1000$, where $\text{avg}(fee_{ad})$ is the average of $fee_{ad}$), click-through rate (CTR), and post-click conversion rate (CVR).
RL-MPCA outperforms all other approaches, and ES-MPCA and DCAF take second and third place, respectively.


 \begin{table}[htbp]
  \centering
  \begin{tabular}{lccccc}\toprule
              & $cost$ & $GPM$ & $CPM$ & $CTR$ & $CVR$ \\
              \midrule
     Static    & +0.00\% & +0.00\% & +0.00\% & +0.00\% & +0.00\% \\
     DCAF     & -0.77\% & +1.38\% & 0.01\% & +0.26\% & +0.53\% \\
     ES-MPCA      & -0.3\% & +2.25\% & +0.12\% & +0.83\% & +0.89\% \\
     \textbf{RL-MPCA}  & -1.5\% & \textbf{+3.68\%} & \textbf{+0.90\%} & \textbf{+1.09\%} & \textbf{+2.86\%}
     \\\bottomrule
  \end{tabular}
  \caption{The online A/B test results.}
  \label{tab:online-experiments-results}
 \end{table}

\section{Conclusion and Future Work}
This paper proposes a Reinforcement Learning based Multi-Phase Computation Allocation approach, RL-MPCA, for recommender systems. RL-MPCA creatively formulates the computation resource (CR) allocation problem as a Weakly Coupled MDP problem and solves it with an RL-based approach. 
Besides, RL-MPCA designs a novel multi-scenario compatible Q-network adapting to various CR allocation scenarios, and calibrates Q-value by introducing multiple adaptive Lagrange multipliers (adaptive-$\lambda$) to avoid violating the global CR constraints when maximizing the business revenue. Both offline experiments and online A/B tests validate the effectiveness of our proposed RL-MPCA approach.

In future work, we plan to explore more general CR allocation approaches and more CR allocation application scenarios. 
Moreover, we plan to explore a new simulation scheme to capture the stochastic variation of response time and system load and then jointly model the response time constraint and the CR constraint to improve the system's availability. 
%


\bibliographystyle{ACM-Reference-Format}
\bibliography{RL-MPCA}

\clearpage

\appendix

\section{Simulation System \label{sec:simulation-system-appendix}}
The offline simulation system contains two modules: the request simulation module and the revenue estimation module. For a given request, the request simulation module is responsible for interacting with an agent and generating interaction results. The revenue estimation module is a deep neural network model based on supervised learning, which evaluates the simulation results and predicts the user views, clicks, and purchases for each request.
Although the offline simulation system requires a lot of time and computation resources, the prediction results of the revenue estimation module are relatively accurate because the request simulation module can generate detailed information about the requests. 
Finally, after calibrating the output of the revenue estimation model, our offline simulation system can achieve fairly confident revenue estimation results. 

As Figure \ref{fig:simulation-system-structure} shows, for each request $i$, the interaction of the simulation system and the agent involves multiple steps.
\begin{figure}[htbp]
  \centering
  \vspace{-0.3cm}
  \includegraphics[width=0.9\linewidth ]{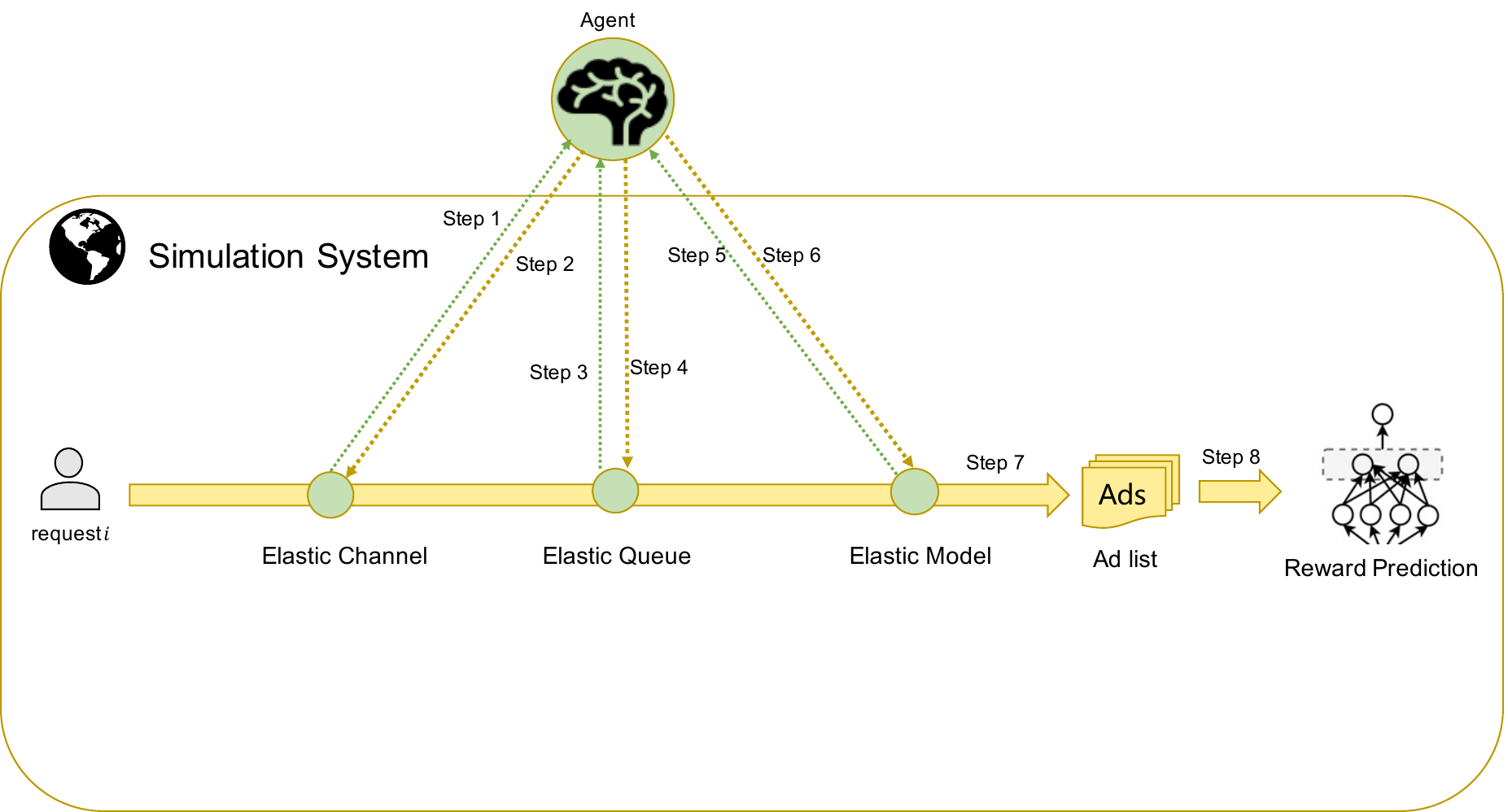}
  \caption{The structure of simulation system.}
  \Description{}
  \label{fig:simulation-system-structure}
\end{figure}
\begin{itemize}
  \item \textbf{Step 1}. The simulation system constructs and feeds the initial state $s_1^i$ to the agent.
  \item \textbf{Step 2}. The agent takes Elastic Channel action $a_1^i$ based on state $s_1^i$.
  \item \textbf{Step 3}. The simulation system retrieves the ads with action $a_1^i$, and feeds state $s_2^i$ (including the retrieval ad list) to the agent.
  \item \textbf{Step 4}. The agent takes Elastic Queue action $a_2^i$ based on state $s_2^i$.
  \item \textbf{Step 5}. The simulation system simulates the truncation operation with the truncation length corresponding to action $a_2^i$, and feeds state $s_3^i$ (including the truncated ad list) to the agent.
  \item \textbf{Step 6}. The agent takes Elastic Model action $a_3^i$ based on state $s_3^i$.
  \item \textbf{Step 7}. The simulation system provides the prediction service for ads with the prediction model corresponding to action $a_3^i$, and outputs state $s_4^i$ (including the truncated ad list and its prediction scores).
  \item \textbf{Step 8}. The simulation system takes state $s_4^i$ as input features, and predicts the final revenue (i.e., user views, clicks, and purchases) with a supervised learning based deep neural network model (see the architecture in Figure \ref{fig:revenue-prediction-model}). 
\end{itemize}
\begin{figure}[htbp]
  \centering
  \includegraphics[width=0.90\linewidth]{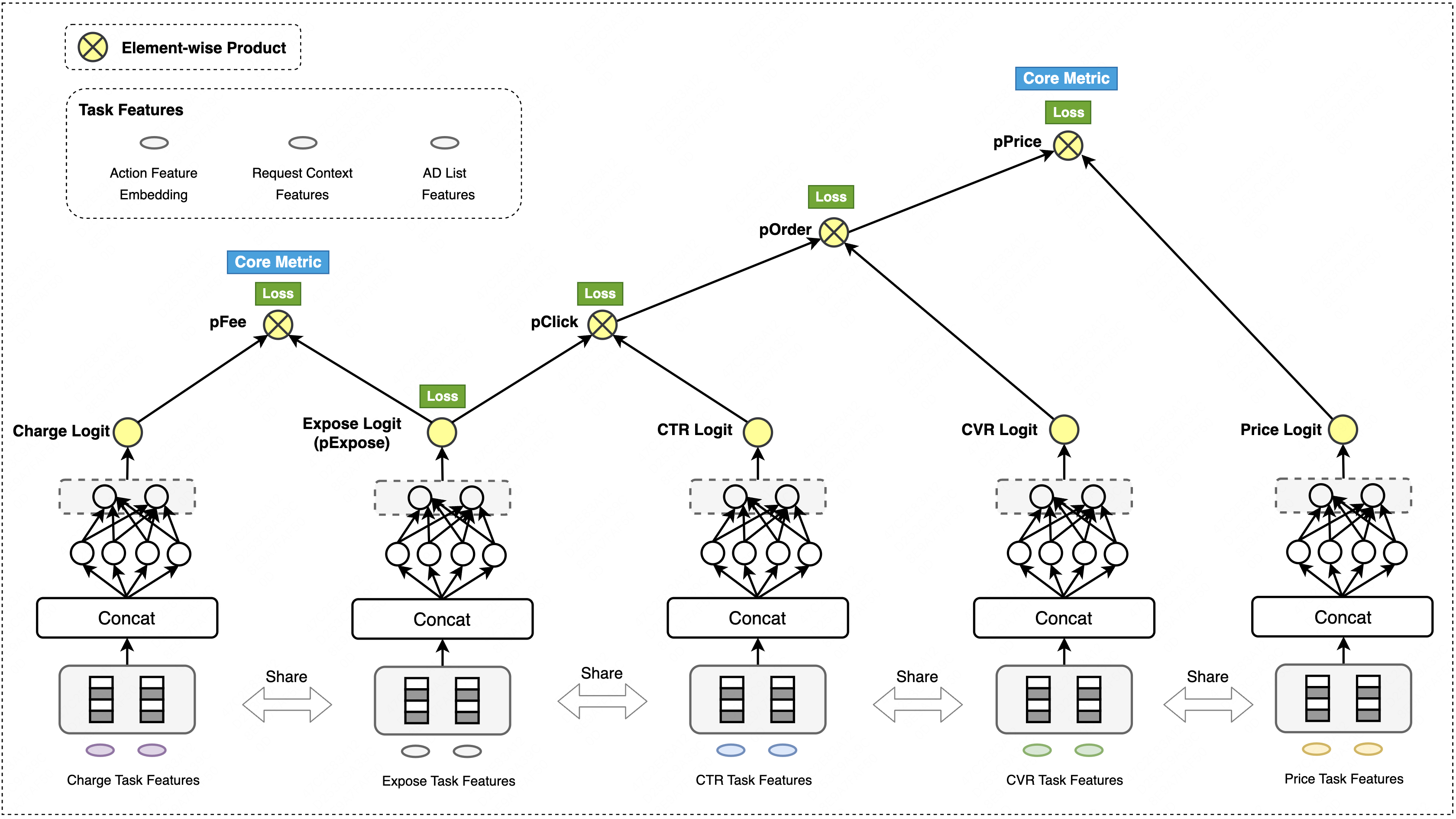}
  \caption{Architecture of revenue prediction model.}
  \Description{}
  \label{fig:revenue-prediction-model}
\end{figure}

\section{Proof \label{sec:proof}}
To slove the single-phase computation resource (CR) allocation problem in Section \ref{sec:constraint-layer}, we introduce a Lagrange multiplier $\lambda_t$, and construct the dual problem:
\begin{align}
  \begin{split}
  &\min_{\lambda_t} \max_{a_t}\sum_{i=1}^M \sum_{a_t=1}^{N_t}x_{i,a_t} {Value}_{i,a_t} \\ &\qquad \qquad \qquad \qquad - \lambda_t \left(\sum_{i=1}^M \sum_{a_t=1}^{N_t} x_{i,a_t} {Cost}_{i,a_t} - C_t \right) \\
  \end{split} \\
  &\qquad \qquad \qquad  s.t.  \ \  \sum_{a_t=1}^{N_t} x_{i,a_t} \leq 1, \ \  \forall i,t \\
  &\qquad \qquad \qquad \qquad  \ \  x_{i,a_t} \in \{0, 1\}, \ \  \forall i,a_t \\ \label{formula:proof-indicator}
  & \qquad \qquad \qquad \qquad \qquad \qquad   \lambda_t \geq 0
\end{align}
In phase $t$, for request $i$, there is one and only one action $a_t$ can be taken. Then the dual problem above can be further transformed as:
\begin{align}
  &\min_{\lambda_t} \sum_{i=1}^M \max_{a_t\in\{1,\dots,N_t\}} \{{Value}_{i,a_t} - \lambda_t {Cost}_{i,a_t}\} + \lambda_t(C_t) \\
  & \qquad \qquad \qquad \qquad \qquad s.t \ \    \lambda_t \geq 0
\end{align}
Thus, we have the global optimal solution to original problem, $x_{i,a_t^*} = 1$ when:
\begin{align}
  a_t^* = \mathop{\arg \max}_{a_t}({Value}_{i,a_t} - \lambda_t Cost_{i, a_t}) \label{formula:optimal-action-for-value-cost-appendix}
\end{align}
Note that a similar proof has been provided in \cite{chen2022bcrlsp}, but the constraint definition of our optimization problem is different from it.
\section{Computation Cost Estimation \label{sec:computation-cost-estimation}} 
Essentially, CRs include computing resources, memory resources, network transmission resources, etc. In real industrial applications, computation cost estimation aims to find a metric that is easy to calculate and can be directly mapped to the amount of computation consumed. CRAS uses \textit{queue length} as the computation cost metric, which is simple and feasible in Elastic Queue scenarios, and we have verified this in Meituan advertising system. 
However, \textit{queue length} does not apply to Elastic Channel and Elastic Model scenarios. Specifically, in Elastic Channel, the primary metric affecting the CR consumption of the retrieval service is the number of requests entering the service. In Elastic Model, the primary metrics affecting the resource consumption of the prediction service are the number of requests and the total number of ads entering the model.
During the model training, we use the number of requests entering the retrieval channel and the number of requests entering the complex prediction model as the computation cost evaluation metrics to facilitate the evaluation of system computation. 
Because the Elastic Queue guarantees the number of ads entering the prediction model, it is reasonable to ignore the number of ads in Elastic Model when training the model.

In the offline experiments and online A/B tests, we also ensured that the number of ads entering the complex prediction model did not exceed the target value.

\section{Hyper-parameters \label{sec:hyper-parameters}}
Table \ref{tab:hyper-parameters settings} lists the hyper-parameters of experiments.
\begin{table}[htbp]
  \centering
  \tabcolsep=0.45cm
  \begin{tabular}{lc}\toprule
              Hyper-parameters & Value  \\
              \midrule
     Adaptive-$\lambda$ update times $K$ & 10 \\
     Learning rate of adaptive-$\lambda$ $\alpha$ & 0.1 \\
     Number of phases & 3 \\
     Sizes of action spaces ($N_c, N_q, N_m$) & (2, 26, 2) \\ 
     Number of heads in the network  & 64 \\
     Size of hidden layer in the network & [128, 64] \\
     Optimizer    & Adam \\
     Learning rate & $3 * 10^{-4}$ \\
     Discount factor $\gamma$ & $0.99$ \\
     Batch size & 8192 \\
     Activation function & ReLU \\
     BCQ threshold $\tau$    & 0.3  \\
     Update frequency of target net $N_{target}$ & 100 \\
     Learning rate of $\lambda$ in Ex-RCPO & $1 * 10^{-4}$ \\
     Temperature coefficient in Ex-CrossDQN & 40 \\
     \bottomrule
  \end{tabular}
  \caption{The hyper-parameters of experiments.}
  \label{tab:hyper-parameters settings}
 \end{table}

\section{ES-MPCA \label{sec:es-mpca}}

Same as RL-MPCA (see more details in Section \ref{sub:weakly-coupled-mdp-problem-formulation}), ES-MPCA also formulates the multi-phase CR allocation problem as a Weakly Coupled MDP problem. The difference is that ES-MPCA solves it with an evolutionary strategies based (ES-based) approach.
To solve the Weakly Coupled MDP problem, we consider it as a black-box optimization problem, aiming to maximize the total business revenue under the CR constraints.

\begin{algorithm}
	\renewcommand{\algorithmicrequire}{\textbf{Input:}}
	\renewcommand{\algorithmicensure}{\textbf{Output:}}
	\caption{Offline Training of ES-MPCA (Based on CEM)}
	\label{alg:offline-training-cem}
	\begin{algorithmic}[1]
		\REQUIRE Number of iteration $I$, the number of parameters $N_{all} = N_{channel} + N_{queue} + N_{model}$, the number of parameters sampled $N_{sample}$, the number of parameters retained $N_{retain}$.
    \STATE Initialize mean $\boldsymbol{\mu}^0 = (\mu_1^0,\dots,\mu_{N_{all}}^0)$ and variance $\boldsymbol{\sigma}^0 = (\sigma_1^0,\dots,\sigma_{N_{all}}^0)$ of parameters.
		\FOR{$i = 1,\dots,I $}
    \STATE Draw sample $\{\boldsymbol{\theta}_1,\dots,\boldsymbol{\theta}_{N_{sample}}\}  \sim N(\boldsymbol{\mu}^{i-1}, \boldsymbol{\sigma}^{i-1})$
		\STATE Evaluate $\{\boldsymbol{\theta}_1,\dots,\boldsymbol{\theta}_{N_{sample}}\}$ by simulation system (Reward Evaluation)
    \STATE Sort $\{\boldsymbol{\theta}_1,\dots,\boldsymbol{\theta}_{N_{sample}}\}$ by the reward $reward = \sum Value(\boldsymbol{\theta}) - \sum_t \lambda_t \min\{C_t - \sum Cost_t(\boldsymbol{\theta}), 0\})$
		\STATE Take top-$N_{retain}$ parameters $\{\boldsymbol{\theta}_1,\dots,\boldsymbol{\theta}_{N_{retain}}\}$, then calculate their mean $\boldsymbol{\mu}^i$ and variance $\boldsymbol{\sigma}^i$
		\ENDFOR
		\ENSURE The best parameter $\boldsymbol{\theta}^* = (\theta_1^*,\dots,\theta_{N_{all}}^*)$
	\end{algorithmic}
\end{algorithm}

In this paper, we use Cross-Entropy Method (CEM) \cite{rubinstein2004cross} to solve the black-box optimization problem. ES-MPCA designs the actions as:
\begin{align}
 channelQuota  &= f_{c}\left(\boldsymbol{\theta}_{c} \boldsymbol{x}_{c}\right) \\
 queueLen  &=f_{q}\left(\boldsymbol{\theta}_{q} \boldsymbol{x}_{q}\right)  \\
 modelQuota  &=f_{m}\left(\boldsymbol{\theta}_{m} \boldsymbol{x}_{m}\right)  
\end{align}
where $channelQuota$, $queueLen$ and $modelQuota$ are retrieval strategy number, truncation length and prediction model number, respectively. $(\boldsymbol{\theta}_{c}, \boldsymbol{\theta}_{q}, \boldsymbol{\theta}_{m})$ and $(\boldsymbol{x}_{c}, \boldsymbol{x}_{q}, \boldsymbol{x}_{m})$ are parameters and features, respectively.

Algorithm \ref{alg:offline-training-cem} describes the training process of CEM-based ES-MPCA. By imposing an extremely large penalty on the parameters that violate the constraint ($\lambda_t$ is generally an extremely large value, e.g., for each phase $t$, $\lambda_t = 10^8$ in our experiments), ES-MPCA always guarantees that the final output optimal parameters $\boldsymbol{\theta}^*$ are those that satisfy the CR constraints. 

Experiment results show that the optimal parameters $\boldsymbol{\theta}^*$ outputted by ES-MPCA always exactly satisfy the CR constraints (i.e., for each phase $t$, $\sum Cost_t(\boldsymbol{\theta}^*) = C_t$ holds), which is consistent with the assumptions and conclusions in Section \ref{sec:constraint-layer}.

\section{Behavioral Policies \label{sec:behavioral-policies}}
In this section, we provide a detailed introduction to behavioral policies. 
Random exploratory policies randomly make decisions in each phase to explore the revenues under different actions, including randomly selecting retrieval channels, truncation lengths, and prediction models. 
Superior policies include ES-based policies and RL-based policies. We train them on a random dataset collected by random exploratory policies. More details of ES-based policies are provided in Appendix \ref{sec:es-mpca}.

\section{Online Serving}

After model training (Algorithm \ref{alg:offline-training}) and $\lambda$-correction (see more details in Section \ref{sec:lambda-correction}), we obtain the trained network $Q_{\theta}$ and trained constraint parameter $\boldsymbol{\lambda}=(\lambda_1,\dots,\lambda_T)$. Algorithm \ref{alg:online-serving} shows the process of online serving for a given request.

\begin{algorithm}
	\renewcommand{\algorithmicrequire}{\textbf{Input:}}
	\renewcommand{\algorithmicensure}{\textbf{Output:}}
	\caption{Online Serving of RL-MPCA}
	\label{alg:online-serving}
	\begin{algorithmic}[1]
		\REQUIRE Trained Network $Q_{\theta}$, trained constraint parameter $\boldsymbol{\lambda}=(\lambda_1,\dots,\lambda_T)$
    \STATE Initialize state $s_0$
    \FOR{$t = 1,\dots,T $}
		\STATE Take action $a_{t}^{*} = \mathop{\arg \max}_{a_{t}} (Q_{\theta}(s_{t}, a_{t})-\lambda_{t} Cost(s_{t}, a_{t}))$
    \STATE Execute allocation following $a_t^*$
    \STATE Observe the next state from system
		\ENDFOR
		\end{algorithmic}  
\end{algorithm}

\end{document}